\documentclass[submission,copyright,creativecommons]{eptcs}
\providecommand{\event}{NCMA 2024} 

\usepackage{graphicx}
\usepackage{amsmath,amssymb}
\usepackage{pifont}
\usepackage{color}
\ifpdf
  \usepackage{underscore}         
  \usepackage[T1]{fontenc}        
\else
  \usepackage{breakurl}           
\fi

\newtheorem{theorem}{Theorem}
\newtheorem{lemma}{Lemma}
\newtheorem{example}{Example}

\newtheorem{proposition}{Proposition}

\newenvironment{proof}{\begin{trivlist}
                       \item[]{\bf Proof}
                       \hspace{0cm} }{\hfill {\large $\bullet$}
                       \end{trivlist}}

\newcommand{\be}{\begin{eqnarray}}
\newcommand{\ee}{\end{eqnarray}}

\newcommand{\ignore}[1]{}

\newcommand{\qed}{}

\title{$5'\rightarrow 3'$  Watson-Crick Automata accepting Necklaces}
\author{Benedek Nagy
\institute{Department of Mathematics,
Eastern Mediterranean University\\
99628 Famagusta, North Cyprus, Mersin-10, Turkey\\
Department of Computer Science, Institute of Mathematics and Informatics,\\ Eszterh\'azy K\'aroly Catholic University, Eger, Hungary}
\email{nbenedek.inf@gmail.com}
}

\def\titlerunning{$5'\rightarrow 3'$  WK Automata for Necklaces}
\def\authorrunning{Benedek Nagy}
\begin{document}

\maketitle

\begin{abstract}
Watson-Crick (WK) finite automata work on a Watson-Crick tape representing a DNA molecule. They have two reading heads. In $5'\rightarrow 3'$  WK automata, the heads move and read the input in opposite physical directions. In this paper, we consider such inputs which are necklaces, i.e., they represent circular DNA molecules. In sensing  $5'\rightarrow 3'$  WK automata, the computation on the input is finished when the heads meet. As the original model is capable of accepting the linear context-free languages, the necklace languages we are investigating here have strong relations to that class.
Here, we use these automata in two different acceptance modes. On the one hand, in \textit{weak} acceptance mode the heads are starting nondeterministically at any point of the input, like the necklace is cut at a nondeterministically chosen point), and if the input is accepted, it is in the accepted necklace language. These languages can be seen as the languages obtained from the linear context-free languages by taking their closure under cyclic shift operation.
 On the other hand,
in \textit{strong} acceptance mode, it is required that the input is accepted starting the heads in the computation from every point of the cycle. These languages can be seen as the maximal cyclic shift closed languages included in a linear language. On the other hand, as  it will be shown, they have a kind of locally testable property.
 We present some hierarchy results based on restricted variants of the WK automata, such as stateless or all-final variants.
\end{abstract}


\section{Introduction}
\label{s:intro}
On the one hand, there are numerous new computational paradigms that emerged in the last decades, usually based on or motivated by some natural phenomena \cite{handbook-of-NC}. A number of them are connected to DNA molecules, thus DNA computing has various theoretical \cite{Paun} and various experimental branches (based e.g., on \cite{Adleman}). 
 Both Watson-Crick automata and the theory/combinatorics of circular words (also called necklaces) are belonging to theoretical DNA motivated models.
 On the other hand, as their names already hint, they have strong connections to classical computing theory, including automata and formal languages.  
Watson-Crick automata (abbreviated by the first and last letters of the names of the Nobel prize winner discoverers of the  DNA molecule structure, i.e., WK automata),  were introduced in \cite{Freund} as an automata type model of DNA computing %
\cite{Sempere04,Czeizle}. These automata are interesting both from theoretical aspects of computations and also from their applicability in bioinformatical problems \cite{Sempere18}.
 The DNA molecules, from a computational point of view, can be seen as linear or circular double stranded words over the alphabet of nucleotides, such that the two strands are related by the Watson-Crick complementarity relation (that is, in nature, a bijective pairing relation on the used 4 nucleotides). 
The original models of WK automata work on double-stranded tapes called Watson-Crick tapes that represent (linear) DNA molecules 
and the two read-only heads scanning the two strands in a correlated manner. These automata are closely related to finite automata having two heads.
 From the biological point of view there are some restrictions that could be applied on the model, e.g., on the number of states or on the number of input letters being read in a transition.
Relationships between various restricted classes of the Watson-Crick automata were presented in \cite{Freund,Paun,Kuske}. 
From another important biological motivation, the reverse and $5'\to 3'$ WK automata make more sense: each (linear) DNA strand has its own $5'$ and $3'$ end, where these names come from the position of the carbon atoms in the sugar part to which the next nucleotide can connect by covalent bond.  
 The two strands of a DNA molecule have opposite chemical direction, i.e., the $5'$ end of a strand gives the $3'$ end of the other and vice versa. Thus, if one believes that in these automata a biochemical sensor, an enzyme, may read the strands, then, most probably, the enzyme reads the two strands in the same chemical direction, i.e., from their $5'$ ends to the direction of their $3'$ ends  \cite{Freund,DNA2008,Leupold-NCMA,Leupold}. While the reverse variant of WK automata is essentially the same as the full-reading non-sensing variant of $5'\to 3'$ WK automata \cite{Paun,Leupold-NCMA}, in the sensing version, the computation on an input finishes at latest when the two heads meet. This sensing was taken into account with a rather artificial sensing parameter in \cite{DNA2008,Nagy2013}, while without it in \cite{AFL,NaCo,UCNC,ActaInf}. In \cite{iConcept} specific both-head stepping variants were defined, where both heads move together and read letter by letter the input (till they meet).   
  We should mention here that sensing $5'\to 3'$ WK automata is closely related to other 2-head finite automata models %
  described under various names
   like linear automata \cite{R.Louk}, biautomata \cite{NCMA-bi} or simply 2-head automata \cite{Triangle}, as their class is capable to accept the class of the linear context-free languages. The specific variant shown in \cite{iConcept} is able to accept the so-called even-linear languages \cite{Amar,Sempere}. Other restricted version, namely $5'\to 3'$ WK automata with exactly one state, was investigated in details in \cite{Eger-N}.
Some extensions of the $5'\to 3'$ WK automata were also developed, e.g., jumping $5'\to 3'$ WK automata \cite{jumpWK},  combination with automata with translucent letters \cite{wtl,WKwtl} or $5'\to 3'$ WK transducers \cite{RAIRO-trans}.
 
 In this paper, the model of $5'\to 3'$ WK automata is used for languages of necklaces, i.e., sets of circular words. As there are circular DNA molecules, it is of particular interest to investigate these automata and analyze their computational power, etc.
 As usual, we are using linearization of necklaces, i.e., we represent a necklace by the set of (linear) words that are obtained as the conjugate class of any of the words that can represent the necklace. We use two modes of acceptance: a necklace is accepted in the weak mode if any of its conjugates is accepted by the given automaton; and those necklaces are accepted in the strong mode for which each of their conjugates are accepted.
  
 In the next section, we formally define our concepts, and then in Sections \ref{w:WK} and \ref{s:WK} we give a sequence of hierarchy results among the accepted classes of necklace languages including all-final, simple, 1-limited, and stateless $5'\rightarrow 3'$  Watson-Crick automata in case of the weak and strong accepting mode for necklaces, respectively. 
 Conclusions and open questions will close the paper. 
 
 Here, we recall only one of the main results for each of the acceptance modes: 
 \begin{itemize}
 \item a language can be weakly accepted by a $5'\to 3'$ WK automaton if and only if it is the cyclic closure of a linear context-free language.
 \item if a language is strongly accepted by a  $5'\to 3'$ WK automaton, then the language has a kind of locally testable property.  
 \end{itemize}
 
%
\section{Preliminaries}
\label{s:pre}
We assume that the reader is familiar with basic concepts of formal languages and automata, otherwise she or he is referred to \cite{HopUl,Handb}. For any unexplained notions about DNA computing we refer, e.g., to \cite{Paun}. We denote the empty word by $\lambda$, and the sets of positive and nonnegative integers by $\mathbb N$ and $\mathbb N_0$, respectively.

Let $T$ be an \textit{alphabet}, then for any word $w \in T^*$ if $w=uv$, then the word $vu$ is a \textit{conjugate} of $w$, and the set of all conjugates of $w$ is called a \textit{necklace} (or cyclic, or circular word) $w_{\circ}$. The operation by which we can obtain each element of the class is called \textit{cyclic shift}, i.e., the cyclic shift of $aw$ is $wa$ where $a\in T$ and $w\in T^*$. The subsequent application of the cyclic shift operation $cycl$ (at most as many times as the length of the word) obtains each conjugate of the word we start with. 
Periodic properties of circular words were studied in \cite{Hegedus,DM}, where a weak period of a circular word was defined as a period of an element of the conjugate class and a period was a strong period if it was a period of each element of the conjugate class.
 In fact, what we are dealing with is the linearization of the circular words. One can imagine those as words written in a cyclic way joining (i.e., concatenating) the first letter of the word after its last letter,  in this way obtaining the word without a starting and without an ending point. Languages of necklaces are also studied in the literature \cite{Manfred}. In this paper, we use necklaces to model (describe) circular DNA molecules. A \textit{language of necklaces} is represented by the union of the necklaces, i.e., conjugate classes. Obviously, this condition can be translated as follows: a language $L\subset T^*$ is a language of necklaces if for any word $w\in L$ each conjugate of $w=uv$ is also in $L$, i.e. $vu\in L$. Consequently, necklace languages are exactly those languages that are closed under cyclic shift operation, when we apply the operation $cycl$ for a language as follows: $cycl(L) = \{uv ~|~ u,v\in T^*, w=vu\in L\} =\bigcup\limits_{w\in L} w_{\circ}$. Further, the cyclic closure of a class $\mathcal{L}$ of languages is the class of the cyclic closures of the languages in $\mathcal{L} $. 

One class of the Chomsky hierarchy, the class of {linear context-free languages}, has a strong connection to the automata model we start with, thus we recall it briefly.
A generative grammar $G=(N,T,S,P)$ is \textit{linear context-free} if every production is context-free and contains at most one nonterminal on the right hand side, i.e., it is one of the forms $A\to u, \ A\to uBv$ with $A,B\in N$ and $u,v\in T^{*}$. A language $L$ is linear context-free if it can be generated by a linear context-free grammar.
This class of languages is denoted by $\mathcal{L}_{LIN}$.
 It is known that, on the one hand, the classes of regular and context-free languages (denoted by $\mathcal{L}_{context-free}$) are closed under cyclic shift,  
 on the other hand, the class of linear context-free languages is not \cite{genCycl,HopUl}. %

The two strands of the DNA molecule have opposite $5'\rightarrow3'$ orientations. Therefore,  Watson-Crick finite automata that parse the two strands of the Watson-Crick tape in opposite directions are investigated. Now, we use them to accept necklace languages.
Figure \ref{sd-init} indicates the initial configuration of such an automaton. As there is no specific start and end point of a necklace, the starting point can be chosen arbitrarily (and based on that we will define two types of acceptance conditions).

\begin{figure}[htb]
    \centering
        \includegraphics[scale=0.75]{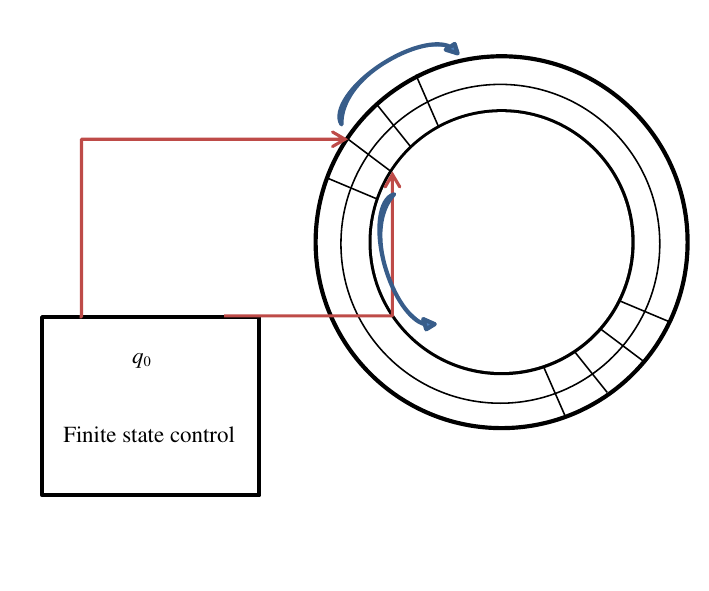}%
            \caption{A sensing $5'\rightarrow 3'$ WK automaton in the initial configuration.} 
  \label{sd-init}
\end{figure}        

A $5'\rightarrow3'$ WK automaton is called \textit{sensing} if it senses that its heads are meeting, i.e., they are in the same position.
As in these models the full input is already processed at that time (if the heads meet again), we use the model to make the decision of the type of the computation at that point, i.e., if the computation is an accepting computation.
  
Formally, a \textit{Watson-Crick automaton} is a 6-tuple $M=(V,\rho,Q,q_0,F,\delta)$, where:
\begin{itemize}
\item $V$ is the (input or tape) alphabet,
\item $\rho\subseteq V\times V$ denotes a complementarity relation,
\item $Q$ represents a finite set of states,
\item $q_0\in Q$ is the initial state,
\item $F\subseteq Q$ is the set of final (also called accepting) states and
\item $\delta$ is called the transition function and it is of the form  $\delta: Q \times \left(\begin{array}{c}V^{*}\\ V^{*}\end{array}\right)\rightarrow 2^Q$, such that it is non-empty only for finitely many triplets $(q,u,v), q \in Q, u,v\in V^*$ when these triplets may also be written either in the form $(q,u,v)$ or $(q,\binom{u}{v})$ indicating which of the strings are read by which of the heads. (The heads can be called upper (left or first) and lower (right or second) heads, respectively.
    \end{itemize}

Based on our definition,
in these WK automata every pair of positions in the Watson-Crick tape is read by exactly one of the heads in an accepting computation, thus the complementarity relation cannot play importance; instead, in this paper we always assume that it is the identity relation.  We are presenting the sensing $5'\to 3'$ WK automaton in Figures \ref{sd-init} and \ref{sd} working on the 2-strand necklace. However,
for the above reason, it is more convenient to consider the input as a ``normal'' necklace and not a double stranded necklace. Actually, this is a usual trick to simplify the notation, as, in some cases, also instead of the nucleotide pairs, e.g., $\left[ C \atop G \right]$ (with $C,G\in V$) one may simply write $a \in T$, by shifting the description to a new alphabet, which can be done always if the complementarity relation is symmetric and bijective (as in the case of real DNA). Thus, we may use alphabet $T$ instead of using $V$ and $\rho$, to \textit{simplify the writing of $5'\to 3'$ WK automaton} to a 5-tuple $M=(T,Q,q_0,F,\delta)$, modifying $\delta$ appropriately to use $T$.  
On the other hand, the complementarity relation can always be replaced by the identity, even in the traditional models, as was proven in \cite{Kuske}.
                  \begin{figure}[htb]
    \centering
                  \includegraphics[scale=0.75]{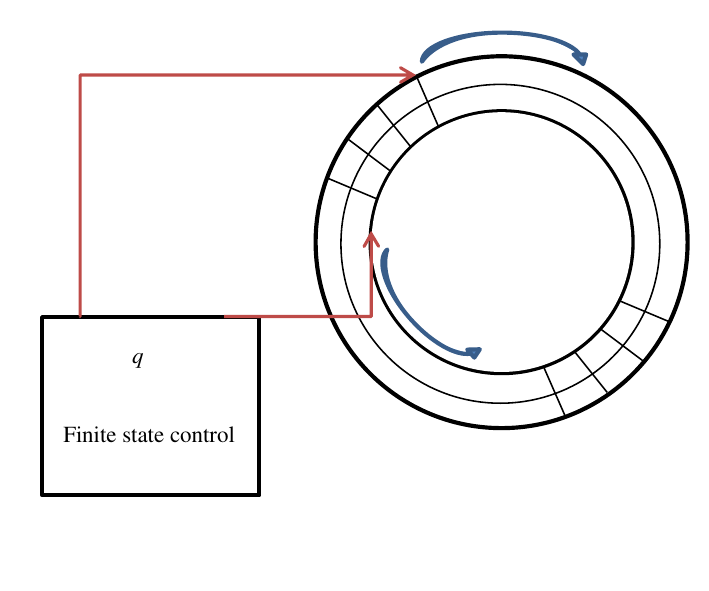}\\ \vspace{-50pt}
          \includegraphics[scale=0.75]{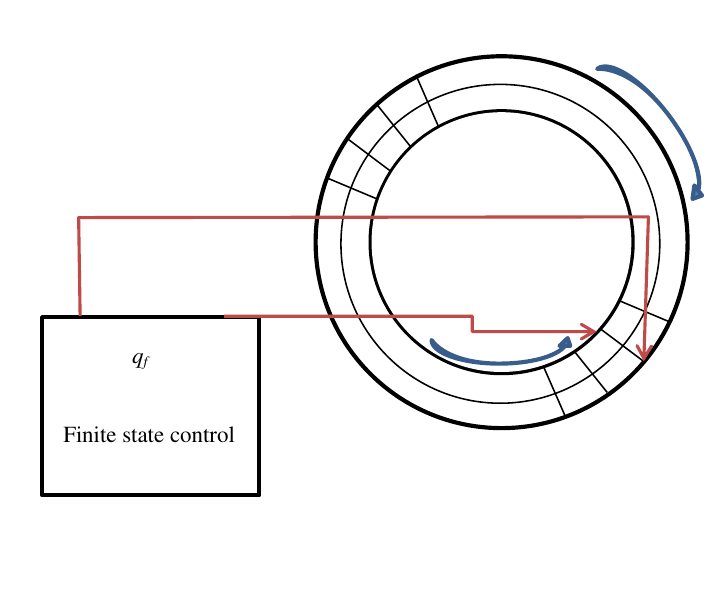}
            \caption{A sensing $5'\rightarrow 3'$ WK automaton in %
            a configuration during a computation %
            and in an accepting configuration with a final state $q_f$ (bottom).}
    \label{sd}
\end{figure}

By continuing the formal description, we consider the \textit{computation} of $5'\to 3'$ WK automata on necklaces as finite sequences of configurations.
A \textit{configuration} %
is a pair $(q,w)$ where $q\in Q$ is the current state of the automaton and $w$ is the part of the input necklace which has not been read (processed) yet written as a normal word as we detail it. 
 In the initial configuration, the initial state $q_0$ is used with any element of the conjugate class of the necklace, mimicking the arbitrary (nondeterministic) choice of a position of the cycle from where the computation starts: the conjugate starting at that position will be processed.  
As the 2 heads are moving in opposite physical directions, the unprocessed part between them will be shorter and shorter until the heads meet (i.e., they are both in the same position again, as this is shown in Figure \ref{sd}).  
 Formally, let  $w',x,y\in T^*,\ q,q'\in Q$. Then, there is a computation step between two configurations:
$(q,xw'y)\Rightarrow (q',w' )$ if and only if $q'\in \delta(q,x,y)$. The reflexive and transitive closure of the relation $\Rightarrow$ is, as usual, denoted by $\Rightarrow^*$ and called computation. For a given conjugate $w\in T^*$  of the input, an accepting computation is a sequence of transitions $(q_0,w) \Rightarrow^* (q_F,\lambda)$, starting from an initial configuration and ending in a configuration consisting of a final state and the empty word.
 Now, based on the conjugate class of a necklace we define our acceptance conditions:
 \begin{enumerate}
\item A necklace $w_{\circ}$ is \textit{weakly accepted} by a WK automaton $M$ if there is a conjugate $vu$ (when $w=uv$, i.e., in this case $vu\in w_{\circ}$) such that there is an accepting computation on $vu$.
\item A necklace $w_{\circ}$ is \textit{strongly accepted} by a WK automaton $M$ if there is an accepting computation for every conjugate of $w$ (i.e., each element of $ w_{\circ}$).
\end{enumerate} The weak and strong distinction comes in a similar manner as it was used for periods in \cite{DM}. 
Now, one may consider the former case, as there is a nondeterministic choice for where to cut the necklace to start the computation, and if this (nondeterministically chosen) starting point leads to an accepting computation, then the necklace is accepted.
Contrariwise, in the latter case, there must be an accepting computation for each possible starting point for a necklace to be included in the accepted language.
 \begin{enumerate}
\item The language $L$ of necklaces is weakly accepted by a WK automaton $M$ if for each word $w\in L$, there is a conjugate $vu$ (when $w=uv$) such that there is an accepting computation on $vu$.
\item The language $L$ of necklaces is strongly accepted by a WK automaton $M$ if for each necklace $w_{\circ} \subset L$, all conjugates of $w$ (i.e., each element of $w_{\circ}$) are accepted by $M$ by some computations.
\end{enumerate}
We may also write these conditions more formally and we can also use some special notation for these languages:
\begin{enumerate}
\item $L_w(M)=\{ w \in T^*  ~|~  \text{ there exist } u,v \in T^*,\  q_f\in F \text{ such that } w=uv, (q_0,vu) \Rightarrow^* (q_f,\lambda) \}.$
\item $L_s(M)=\{ w \in T^*  ~|$  for each $ u \in w_{\circ}$ there is a $  q_f\in F $ such that %
 $ (q_0,u) \Rightarrow^* (q_f,\lambda) \}.$
\end{enumerate}
The classes of necklace languages weakly and strongly accepted by sensing $5' \rightarrow 3'$ WK automata are denoted by $\mathcal{W}_*$ and $\mathcal{S}_*$.
There are some restricted variants of WK automata which are usually considered (e.g., from computational and biological motivations):
\begin{itemize}
\item $\textbf{N}$: stateless, i.e., with only one state: if $Q=F=\{q_0\}$;
\item $\textbf{F}$: all-final, i.e., with only final states: if $Q=F$;
\item $\textbf{S}$: simple (at most one head moves in a step) $\delta:(Q\times ((\{\lambda\},T^* )\cup(T^*,\{\lambda\})))\rightarrow 2^Q$.
\item $\textbf{1}$: 1-limited (exactly one letter is being read in each step) $\delta:(Q\times ((\{\lambda\},T)\cup (T,\{\lambda\})))\rightarrow2^Q$.
\end{itemize}
Let $\mathcal{W}_N$, $\mathcal{W}_F$, $\mathcal{W}_S$ and $\mathcal{W}_1$ denote the necklace language classes weakly accepted by sensing {\textbf{N}}, {\textbf{F}}, {\textbf{S}} and {\textbf{1}} $5'\rightarrow 3'$ WK automata, respectively. 
 Further variants having multiple constraints can also be defined as sensing \textbf{F1}, \textbf{N1}, \textbf{FS}, \textbf{NS}  $5'\to 3'$ WK automata. Their weakly accepted language classes are denoted by $\mathcal{W}_{F1}$, $\mathcal{W}_{N1}$, $\mathcal{W}_{FS}$ and $\mathcal{W}_{NS}$, respectively.
Similarly, the notation $\mathcal{S}_N, \mathcal{S}_F, \mathcal{S}_S, \mathcal{S}_1$, etc. will be used for the classes that are strongly accepted by the restricted classes of $5'\rightarrow 3'$ WK automata, respectively. Further, we may use the traditional way of acceptance for `ordinary' (i.e., not necessarily necklace) languages and we use the notation for these classes, $\mathcal{L}_*, \mathcal{L}_N$, etc., respectively.

%
\section{On weakly accepted necklace language classes}\label{w:WK}

With this section our aim is twofold. On the one hand we would like to present some general result on the class $\mathcal{W}_*$ and, on the other hand, we are presenting hierarchy results among the language classes of necklaces that are weakly accepted by the restricted models.

The next proposition is a direct consequence of the definitions and the fact that exactly the class $\mathcal{L}_{LIN}$ is accepted the class of (unrestricted)  sensing $5'\rightarrow 3'$ WK finite automata \cite{DNA2008,AFL,CiE,Nagy2013,NaCo}.

 \begin{proposition}
 The cyclic closure $cycl(\mathcal{L}_{LIN})$ is weakly accepted by sensing $5'\rightarrow 3'$ WK finite automata, that is, for 
 each linear language $L$, its cyclic closure $cycl(L)$ is in $\mathcal{W}_*$ and for each language $L' \in \mathcal{W}_*$ there is a linear language $L''$ such that $L' = cycl(L'')$.   
 
Moreover, for each restricted class \textbf{x}$ \in \{$\textbf{S},\textbf{1},\textbf{F},\textbf{N},\textbf{FS},\textbf{F1},\textbf{NS},\textbf{N1}$\}$, $$\mathcal{W}_x = cycl(\mathcal{L}_x) ,$$ i.e., 
the class of weakly accepted necklace languages by a restricted class of sensing $5'\rightarrow 3'$ WK finite automata is the same as
    the cyclic closure of the languages accepted by the class of sensing $5'\rightarrow 3'$ WK finite automata with the same  restriction.
  \end{proposition} 
  
  The cyclic closure of the class $\mathcal{L}_{LIN}$ was also defined as a kind of class of necklace languages (i.e., languages of cyclic words) among many other classes based on a somewhat similar idea in \cite{Manfred}.

As $\mathcal{L}_{LIN}$ is not closed, but $\mathcal{L}_{context-free}$ is closed under cyclic shift \cite{genCycl,HopUl}, we can relate our classes to the Chomsky hierarchy. As, clearly both  $\mathcal{L}_{LIN}$ and  $\mathcal{L}_{context-free}$ contain some languages that are not necklace languages, (e.g., the singleton language $\{ab\}$), we have:
\begin{proposition}
The inclusion  $ \mathcal{W}_* \subsetneq \mathcal{L}_{context-free}$ is proper, while the classes
 $ \mathcal{W}_*$ and $\mathcal{L}_{LIN}$ are incomparable under set theoretic inclusion.
\end{proposition}

Now we show some equivalences among the classes on the top of the hierarchy. 
By, e.g., \cite{AFL,NaCo},
 it is known that sensing $5'\rightarrow 3'$ WK finite automata
accept exactly the linear context-free languages, moreover the same class is accepted by the classes of the following variants: 
  $\mathcal{L}_{LIN}=\mathcal{L}_{*}=\mathcal{L}_{S}=\mathcal{L}_{1}.$ 
This gives the consequence that the weakly accepted classes will also be the same:
\begin{proposition} The following classes of necklace languages are identical: 
  $cycl(\mathcal{L}_{LIN})=\mathcal{W}_{*}=\mathcal{W}_{S}=\mathcal{W}_{1}.$
\end{proposition}

In the rest of the section we present various hierarchy results of the considered necklace languages.

To show that none of he language classes is empty, we start with the most restricted class, the necklace languages weakly accepted by sensing \textbf{N1} $5'\to 3'$ WK automata, to give an example language.

\begin{proposition}
The language $L=\{1^i0^j1^k~|~ i,j,k \in \mathbb N_0 \} \cup \{0^i1^j0^k~|~i,j,k \in \mathbb N_0\}$  is weakly accepted by the sensing  \textbf{N1} $5'\rightarrow3'$ WK automaton: $M = (\{0,1\},\{q\},q,\{q\},\delta)$ with two allowed transitions $q\in \delta(q,0,\lambda)$ and $q\in \delta(q,\lambda,1)$. 
 \end{proposition} \begin{proof}
Clearly the automaton has only one state and it reads exactly one letter in each step of the computation, thus it is
a sensing \textbf{N1} $5'\to 3'$ WK automaton. 
 
 Now, considering the accepted language, for each word of the language there is a conjugate in the  form $0^n 1^m$ (for $w=1^i0^j1^k$, $n=j$ and $m=i+k$; for $w=0^i1^j0^k$, $n=i+k$ and $m=j$). On the other hand, $M$ is accepting
 the language $L(M) = \{0^n1^m~|~n,m\in\mathbb N_0\}$ when the first transition is used $n$, the second one $m$ times during the computation. Now, as $cycl(L(M)) = L$, the language $L$ is  weakly accepted by $M$. The proof is complete. 
 \end{proof}
On the one hand,
as all sensing $\textbf{N1}$ $5'\to 3'$ WK automata are also sensing $\textbf{F1}$ and also sensing $\textbf{NS}$ $5'\to 3'$ WK automata, we have obvious inclusion among the (weakly) accepted language classes. On the other hand, we state and shall prove that both of these inclusions are proper.

\begin{theorem}\label{thm:2}
Each of the classes $\mathcal{W}_{F1}$ %
 and %
  $\mathcal{W}_{NS}$ properly includes the class $\mathcal{W}_{N1}$:
	$$\mathcal{W}_{N1}\subsetneq \mathcal{W}_{F1} \text{ \ \  and \  \ } \mathcal{W}_{N1}\subsetneq \mathcal{W}_{NS}.$$ 
\end{theorem}
\begin{proof} As the inclusions are obvious by definition, we shall prove only their properness.

Let us consider the first statement and the language $L=\{1^i0^j1^k~|~ i,j,k \in \mathbb N_0, \ j\in\{ i+k, i+k+1  \} \} \cup \{0^i1^j0^k~|~i,j,k \in \mathbb N_0, \ i+k\in\{j,j+1 \}\}$.
It can  weakly  be  accepted by a sensing \textbf{F1}   $5' \rightarrow 3'$ WK automaton: Let
   $M = (\{0,1\},\{q,p\},q,\{q,p\},\delta)$ with two allowed transitions $p\in \delta(q,0,\lambda)$ and $q\in \delta(p,\lambda,1)$. 
 Notice that for each word in $L$ there is a conjugate in the form $0^n 1^m$ with the condition that either $n=m$ or $n=m+1$. However, $M$ is accepting exactly the language $L' = \{0^n 1^m~|~ n\in\{m,m+1\} \}$, and thus weakly accepting $L = cycl(L')$. 
 
 To complete the proof of the first part, we should show that %
 $L$ cannot be weakly accepted by any sensing {\textbf{N1}} $5' \rightarrow 3'$ WK automata.
This part of the proof goes by contradiction:
Suppose that $L$ is weakly accepted by a sensing \textbf{N1}  $5' \rightarrow 3'$ WK automaton $M'$ with the sole state $r$ and transition mapping $\delta'$. As $0 = 0^11^0 \in L$, $M'$ must have at least one of the loop-transitions $r\in \delta'(r,0,\lambda)$ and $r\in\delta'(r,\lambda,0)$. However, in either case, all words (and thus all necklaces) of $0^*$ would be (weakly) accepted. However, no words (necklaces) of $0^*$ other than $\lambda$ and $0$ are in the language.
This contradiction proves that there is no sensing \textbf{N1} $5' \rightarrow 3'$ WK automaton that weakly accepts $L$, thus, the first statement of the theorem has been proven.

Considering the second statement, let us consider the language $L''=\{1^i0^j1^k~|~ i,j,k \in \mathbb N_0, j \text{ is even}\} \cup \{0^i1^j0^k~|~i,j,k \in \mathbb N_0, i+k \text{ is even}\}$.
On the one hand, we show that $L''$ is weakly accepted by a sensing \textbf{NS} $5'\to 3'$ WK automaton.
 Thus, let $M'' = (\{0,1\},\{q\},q,\{q\},\delta'')$ with two allowed transitions $q\in \delta''(q,00,\lambda)$ and $q\in \delta''(q,\lambda,1)$. Then the language accepted by $M''$ is $L(M'') = \{ 0^{2n}1^m~|~ n,m\in\mathbb N_0  \}$, and its cyclic closure $cycl(L(M'')) = L''$, i.e., there is an even number of $0$s such that either they are next to each other, or they form the prefix and suffix of the word in $L''$.
 Now, on the other hand, we shall prove that $L''$ is not weakly accepted by any sensing \textbf{N1} $5' \rightarrow 3'$ WK automata. To show this, notice that $00 \in L''$, but $0\not\in L''$. However, a sensing \textbf{N1} $5' \rightarrow 3'$ WK automaton must read the input letter by letter, and each already read part must also form an accepted word, thus
to accept $00$, the automaton must read a $0$ in the first step of the computation, however, then $0$ would also be accepted.
In this way the proper inclusion of the second statement has also been proven. \qed
\end{proof}

\begin{theorem}\label{thm:3} Each of
the classes $\mathcal{W}_{N}$ %
 and %
  $\mathcal{W}_{FS}$ properly includes the class $\mathcal{W}_{NS}$:
$$\mathcal{W}_{NS}\subsetneq \mathcal{W}_{N} \text{ \ \  and \ \  } \mathcal{W}_{NS}\subsetneq \mathcal{W}_{FS}.$$	
\end{theorem}
\begin{proof}
Let us start with the first statement and consider the necklace language $L=\{0^i 1^j 0^k ~|~ \text{ there exist }$ $ n,m \in \mathbb N_0 \text{  such that  } i+k = 2n+m, j = 2m+n   \} \cup \{1^i 0^j 1^k ~|~ \text{ there exist } n,m \in \mathbb N_0 \text{  such that  } i+k = 2n+m, j = 2m+n \}$.
Now, on the one hand, the automaton $M=(\{0,1\},\{q\},q,\{q\},\delta)$ with two transitions $q\in\delta(q,0,11)$ and $q\in\delta(q,00,1)$ is weakly accepting $L$, as each element of $L$ has a conjugate $0^r 1^s $ with $n,m\in\mathbb N_0$ such that $r=n+2m$ and $s=2n+m$, where in fact $n$ and $m$ are the numbers of the computation steps made by the two possible transitions, respectively. Further, it is easy to see that $M$ is a sensing {\textbf{N}} $5' \rightarrow 3'$ WK automaton. Now, on the other hand, it shall be shown that $L$ is not weakly accepted by any sensing {\textbf{NS}} $5' \rightarrow 3'$ WK automaton. This part of the proof is by contradiction, thus let us assume that there is such \textbf{NS} automaton $M'$ that weakly accepts $L$. As at least one of the words $011,101,110$ is accepted by $M'$ (with sole state $p$ and transition mapping $\delta'$) to include this necklace in the language, the automaton $M'$ must have at least  one of the following six transitions:   
$p \in \delta'(p,011,\lambda)$, $p \in \delta'(p,101,\lambda)$, $p \in \delta'(p,110,\lambda)$, $p \in \delta'(p,\lambda,011)$, $p \in \delta'(p,\lambda,101)$ and $p \in \delta'(p,\lambda,110)$. 
However, now by applying the same transition in three consecutive computation steps, it leads to accept the following word:
$011011011$, $101101101$, $110110110$, $011011011$, $101101101$ or $110110110$ respectively to the six cases.
As all these words contain more than two `blocks' of 0's, clearly none of them is in $L$, thus this contradicts to the fact that $M'$ weakly accepts $L$.
Therefore, the language $L$ cannot be weakly accepted by any sensing \textbf{NS} $5'\rightarrow 3'$ WK automata, completing the proof of the first statement.

Now, let us consider the second statement with the witness language 
 $L=\{1^i0^j1^k~|~ i,j,k \in \mathbb N_0, \ j\in\{ i+k, i+k+1  \} \} \cup \{0^i1^j0^k~|~i,j,k \in \mathbb N_0, \ i+k\in\{j,j+1 \}\}$ used in the proof of the previous theorem. Clearly, the sensing {\textbf{F1}}  $5' \rightarrow 3'$ WK automaton given there is also a sensing \textbf{FS}  $5' \rightarrow 3'$ WK automaton. On the other hand, as one needs to have one of the transitions to accept the word $0$ as it is written in the previous proof, every  sensing {\textbf{N}}  $5' \rightarrow 3'$ WK automaton
 must also accept words like $00$ and $000$ which are not in $L$ (and not any of their conjugates are in $L$).  
This contradiction proves the properness of the inclusion in the second statement.
\qed \end{proof}

\begin{theorem}\label{thm:4}
The class $\mathcal{W}_{FS}$ properly includes the class $\mathcal{W}_{F1}$:
 $$\mathcal{W}_{F1}\subsetneq \mathcal{W}_{FS}  .$$
\end{theorem}
\begin{proof}
	As the inclusion is trivial by definition, we need to show only its properness. Let us consider the witness language $L$ defined by the regular expression $(11)^*$. $L$ contains all words over the unary alphabet with even length.
	Now, on the one hand, let  $M=(\{1\},\{p\},p,\{p\},\delta)$ be a sensing  \textbf{FS} $5' \rightarrow 3'$ WK automaton (in fact also \textbf{NS} and \textbf{N}) with the only transition $p\in\delta(p,11,\lambda)$. Clearly, $L(M) = L_w(M) = L$.
 
On the other hand, we need to show that no sensing  \textbf{F1} $5' \rightarrow 3'$ WK automata can (weakly) accept $L$.		As all states of \textbf{F1} automata are final and they should read the input letter by letter, there must be a configuration when only $1$ is read in the accepting computation of, e.g., $11$. As the state of this configuration must also be accepting, $1$ is also accepted (and weakly accepted) by any \textbf{F1} automata that are able to accept $11$. As $1\not\in L$, this leads to a contradiction, thus there is no  sensing  \textbf{F1} $5' \rightarrow 3'$ WK automata that weakly accept $L$.
\end{proof}

\begin{theorem}
The class $\mathcal{W}_{F}$ properly includes the class $\mathcal{W}_{N}$:
$$\mathcal{W}_{N}\subsetneq \mathcal{W}_{F}.$$
\end{theorem}
\begin{proof} On the one hand, the inclusion is trivial by definition. On the other hand, for the properness, let us consider the witness language the regular language
$L= 0^* + 0^*10^*$ which is also a necklace language.
 Let $M=(\{0,1\},\{p,q\},p,\{p,q\},\delta)$ with three transitions $p\in\delta(p,0,\lambda)$, $q\in\delta(p,1,\lambda)$ and $q\in\delta(q,0,\lambda)$, then $L(M) = L_w(M) = L$, moreover $M$ is a sensing \textbf{F}  $5' \rightarrow 3'$ WK automaton (in fact it is also \textbf{FS} and \textbf{F1}).
 To show the properness, we need to show that there is no  sensing \textbf{N}  $5' \rightarrow 3'$ WK automaton that weakly accepts $L$. The proof goes by contradiction, thus let us assume that $M'$ is an automaton with its sole state $q'$ and transition mapping $\delta'$ such that $L_w(M') = L$.
As the word (and also a necklace) $1$ is accepted, $M'$ must have at least one of the transitions $q'\in\delta'(q',1,\lambda)$ and  $q'\in\delta'(q',\lambda,1)$. However, in either case, the words (and necklaces) $11$ and $111$ are also accepted.
However, as they are not in $L$, we have reached a contradiction. This contradiction shows that $L$ is not weakly accepted by any sensing \textbf{N} $5' \rightarrow 3'$ WK automata, and the proof is complete. \qed
\end{proof}

\begin{theorem}\label{thm:5}
The class $\mathcal{W}_{F}$ properly includes the class $\mathcal{W}_{FS}$:
$$\mathcal{W}_{FS} \subsetneq \mathcal{W}_{F}.$$
\end{theorem}
\begin{proof} %
We need to show only the properness, thus let us have the witness language $L = \{0^i 1^j 0^k ~|~ 
j=  i+k \} \cup \{1^i 0^j 1^k ~|~  i+k = j \}$. $L$ is the cyclic closure of the linear context-free language $\{0^{n} 1^{n}~|~ n\in \mathbb N_0 \} $. Clearly, as the automaton $(\{0,1\},\{q\},q,\{q\},\delta)$ with a sole transition $q\in\delta(q,0,1)$ accepts the above mentioned linear context-free language, it also weakly accepts $L$. This automaton is a sensing \textbf{N} $5' \rightarrow 3'$ WK automaton, and thus, it is also  a sensing \textbf{F} $5' \rightarrow 3'$ WK automaton. 
Thus, we need to show only that $L$ cannot be weakly accepted by any sensing \textbf{FS} $5' \rightarrow 3'$ WK automaton. The proof is by contradiction, thus let us assume that $M$ is a sensing  \textbf{FS} $5' \rightarrow 3'$ WK automaton such that $L_w(M) = L$. 
            For each WK automaton, as its transition function gives nonempty sets only for finitely many triplets, there is a maximal length of strings that can be read in a computation step. Let $r$ be this maximal length for automaton $M$.
 Let us consider the word $w = 0^{3r}1^{3r} \in L$.
   Since the length of $w$ is large, $M$ needs more than three computation steps to accept one of its conjugates,
   let us say $u=0^{i}1^{3r}0^{3r-i}$ (or symmetrically, $v=1^{j}0^{3r}1^{3r-j}$; in this latter case the proof is analogous to the case we present here for $u$).
   Now, on the one hand, as $M$ is  \textbf{S} $5' \rightarrow 3'$ WK automaton, exactly one of the heads can move in each computation step, thus always a prefix or a suffix of the (remaining) input is processed (and as the input must be processed, there must be computation steps by reading the input).
  On the other hand, $M$ is also   \textbf{F} $5' \rightarrow 3'$ WK automaton, thus any computation step leads to the acceptance of the word composed by the already read prefix and suffix of the input.
   Therefore, there are two cases. 
   
   If the prefix, let us say $x$ is read in the first step (when input letter is processed), then the prefix $x$ of $u$ must also be in $L$, thus it must also contain at least one occurrence of $0$s and also of $1$s: $x=0^i 1^i$ ($i<r$) must hold, and the remaining input is $1^{3r-i}0^{3r-i}$ (where $3r-i > 2r$).  Now, in the next step (of the accepting computation of $u$ when some input letters are processed) again a prefix or a suffix of the remaining input is read, however, both the block of $0$s and $1$s are so large that either only $1$s are read (prefix case) or only $0$s are read (suffix case). Both lead to the acceptance of some words and necklaces where the number of $0$s and $1$s mismatch, and thus this leads to a contradiction.  
   
  In the second case, if the suffix $y$ of $u$ is read in the first computation step (of the accepting computation of $u$, when at least one letter is processed), then as $y$ is accepted by $M$, $y\in L$ must also hold, and thus $y$ must contain also both $0$ and $1$: $y = 1^{3r-i}0^{3r-i}$ (and in this case $i>r$). The remaining input after this step is $0^{i} 1^{i}$.
  Now, by the second step (of the computation consuming input letter(s)), either the prefix or the suffix of this remaining input is read, but with  length at most $r$, meaning that either only $0$s or only $1$s can be read. But this would lead again to an acceptance of a word (and thus to the weak acceptance of a necklace) that has mismatching numbers of $0$s and $1$s. 
This fact contradicts to our assumption, %
hence $L$ cannot be weakly accepted by any sensing \textbf{FS} $5' \rightarrow 3'$ WK automata and thus the proof is complete. %
\qed\end{proof}

Finally, we present our last hierarchy result of the section by showing that all-final automata are weaker than the unrestricted variants in the term of weakly accepting language classes. 
\begin{theorem}\label{thm:6}
The class $\mathcal{W}_{*}$ properly includes the class $\mathcal{W}_{F}$:
 $$\mathcal{W}_{F} \subsetneq \mathcal{W}_{*} .$$
\end{theorem}

\begin{proof} Again, we need to prove only properness. Consider the witness language
 $L=\{0^i 10^n 1 0^j ~\mid~ n\in \mathbb N, i,j\in\mathbb N_0, \ i+j = n\}$.
 As $L$ is the cyclic closure of the linear language $\{0^n 10^n 1 ~\mid~ n\in \mathbb N\}$, it is in $\mathcal{W}_*$.
Now, on the other hand, we show that there is no sensing \textbf{F} $5'\rightarrow3'$ WK automaton which weakly accepts $L$. The proof is by contradiction. Thus, let us assume that the language $L$ is weakly accepted by a sensing \textbf{F} $5' \rightarrow 3'$ WK automaton, say $M$. For each WK automaton, as its transition function gives nonempty sets only for finitely many triplets, there is a maximal length of strings that can be read in a computation step. Let $r$ be this maximal length for automaton $M$. %
Let us consider the necklace $w_{\circ} = (0^{2r}10^{2r}1)_\circ \subset L$. %
In any of the conjugates of $0^{2r}10^{2r}1$, the distance of the two occurrences of $1$s is $2r$ implying that at most one of them can be read in the first step of the computation. However, 
  as $M$ is all-final, each computation step leads to an accepted word, and thus, to a weakly accepted necklace. Therefore, as $L(M)$ must contain a word containing at most one $1$, $L_w(M)$ has a necklace containing less than two occurrences of $1$ which is contradicting to the assumption that $L_w(M) = L$.  
\qed\end{proof}

The hierarchy results of this section will be summarized on a Hasse diagram in the concluding section.

\section{On strongly accepted necklace language classes}
\label{s:WK}

In this section we use the strong acceptance mode, i.e., a necklace is in the accepted language if and only if all of its conjugates are accepted by the automaton. By understanding the acceptance mode, and knowing that sensing $5'\to 3'$ WK automata accept exactly the languages of $\mathcal{L}_{LIN}$ (\cite{DNA2008,Nagy2013}),
we can deduce the following fact.

\begin{proposition}
Let $L \in \mathcal{L}_{LIN}$ be a linear context-free language. The maximal necklace language $L' \subset L$ contains exactly those words (necklaces) for which all conjugates (members) are in $L$. Then there is a sensing $5'\to 3'$ WK automaton $M$ that accepts $L$, further, for this automaton $M$, $L_s(M) = L'$.

Moreover, the statement hold also in the other direction:
Let $M'$ be a $5'\to 3'$ WK automaton. The strongly accepted necklace language $L_s(M')$ is the maximal necklace language $L$ such
that $L\subset L(M')$ holds.
\end{proposition}

 Now we introduce a notion for necklaces. If there is a subword $x$ that occurs in some of the conjugates of $w$, then we say that $x$ is a \textit{pattern} in the necklace $w_{\circ}$. If this pattern can be written as $x = u' v'$, then we say it fits to the necklace in the (cut) point that defines the conjugate $w'$ in $w_{\circ}$ such that $u'$ is suffix and $v'$ is a prefix of $w'$. Actually, we can see that one part of $x$ is the prefix and the rest is the suffix of this conjugate. Notice that depending on the length of the pattern there are usually more than one positions where it fits.
 
   We give an example to help the reader to easily catch the concept.
 \begin{example}
 Let the necklace be defined by the word $abcabcaaacb$. Then we have a pattern $aacba$ in it, as it is a subword of, e.g., the conjugate $bcabcaaacba$ (especially, it is a suffix here).
  Now, this pattern fits to the necklace to any points where it occurs, e.g., if we ``cut'' the necklace to obtain the conjugate
  $cbabcabcaaa$, then $u' = aa  $ and $v' = cba$, thus our pattern is used as $aa \cdot cba$.
 \end{example}        
 
Because of the special acceptance mode, we have a kind of locally testable property of all these languages.
(See \cite{McNau-loc-test,LocTest} for related concepts and language families defined in this way.)

\begin{proposition}\label{pro-pat}
Let $L$ be a necklace language strongly accepted by some $5'\to 3'$ WK automata. 
Then there is a finite set of patterns such that for each position of the necklace at least one of them must fit.  
\end{proposition}
 \begin{proof}
   As, there must be an accepting computation for each conjugate of a word of the language $L$, for every (starting) point, one of the possible transitions from the initial state must match. Let us analyze the case formally.
   Let $w \in L$ (i.e., $w_{\circ} \subset L$). Then for each starting point the computation could start, i.e., for each conjugate $w'$ of $w$, there must be a suffix $u'$ and a prefix $v'$ of $w'$ such that there is a transition with them, i.e., $\delta(q_0, v', u')\ne \emptyset$. That means that the pattern $u' \cdot v'$ %
   fits to this cut point of the necklace.  
   On the one hand, there are finitely many possible transitions from the state $q_0$ giving finitely many patterns. On the other hand, for each position at least 
one of them must match to have an accepting computation for that conjugate.   
 \end{proof}  

In some special cases, e.g., if the heads read the same length subwords in each transition, the relation with some classes of locally testable languages can be more immediate. 

On the other hand, the property stated in the previous proposition must hold for each language in $\mathcal{S}_*$, but for some languages there could be more (meaning more complex) restrictions as we can see later.

Now we turn to present some hierarchy results among the corresponding necklace language classes.
As the very first result in this line, we show that even the most restricted class is not empty, i.e., 
there are languages in $\mathcal{S}_{N1}$. Actually, we show more, we give a full characterization 
of this class.

\begin{theorem}\label{prop:sN1}
A necklace language $L$ is in $\mathcal{S}_{N1}$ if and only if $L = T_1^* \cup T_2^*$ for two %
 alphabets $T_1, T_2$. 
\end{theorem}
\begin{proof}
 The proof goes by two parts. First we show that every language of the form $T_1^* \cup T_2^*$ for two %
  alphabets $T_1$ and $T_2$ is in $\mathcal{S}_{N1}$.
 By considering $T_1 = \{a_1,\dots, a_n\}\cup\{c_1,\dots,c_m\}$ and $T_2=\{b_1,\dots,b_k\}\cup\{c_1,\dots,c_m\}$, let us define the automaton $M = (T,\{q\},q,\{q\},\delta)$
with $\delta(q,x,\lambda) = \{q\} $ for each $x \in T_1$ and 
  $\delta(q,\lambda,y) = \{q\} $ for each $y \in T_2$. (For any other triplets let $\delta$ give the empty set.)
Clearly, $M$ is a sensing \textbf{N1} $5'\to 3'$ WK automaton. Moreover, $M$ accepts $T_1^*$ if only the first head is used during the computation and $T_2^*$ if only the second head is used during the computation.
 Now, we show that there is no necklace that can be accepted such that both heads must be used. 
 Contrary, let us assume that
 there is a necklace $w_{\circ}$ which contains letters from both $T_1 \setminus T_2$ and $T_2 \setminus T_1$, then there is a pattern $a_i b_j$ in $w_{\circ}$, i.e., it has a conjugate $b_j u a_i$ (with some $u\in (T_1\cup T_2)^*$). However, there is no transition defined in $M$ to start the computation for this conjugate, thus this necklace cannot be accepted.  
Finally, as $T_1^* \cup T_2^*$ is a necklace language itself, the maximal necklace language in it is also itself, thus
$M$ accepts the necklace language $T_1^* \cup T_2^*$ in strong acceptance mode.

Actually, every sensing \textbf{N1} $5'\to 3'$ WK automaton can be described by two (maybe not disjoint) sets $T_1$ and $T_2$ of letters
having transitions  $\delta(q,a_i,\lambda) = \{q\} $ for each $a_i\in T_1$ and 
  $\delta(q,\lambda,b_j) = \{q\} $ for each $b_j\in T_2$.
  Then, with a similar argument as we used above, one can see that the language $T_1^* \cup T_2^*$ is accepted, and actually, for each accepted word there is a computation where only one of the heads is used to read the entire input. No input can be accepted that has letters that cannot be read by the same head.
\end{proof}

Now, we present some hierarchy results among various classes of strictly accepted necklace languages.

\begin{theorem}
The class $\mathcal{S}_{NS}$ properly includes the class $\mathcal{S}_{N1}$: $$\mathcal{S}_{N1} \subsetneq \mathcal{S}_{NS} .$$
\end{theorem}
\begin{proof} The inclusion holds by definition, as all \textbf{N1} automata are also \textbf{NS} automata. To show the properness we give an example.
Consider $M = (\{a\},\{q\},q,\{q\},\delta)$
with $\delta(q,aa,\lambda) = \{q\} $ and 
   $\delta$ gives the empty set for any other triplets.
It is easy to see that both the accepted and the strongly accepted language is $(aa)^*$ which cannot be accepted by any \textbf{N1} $5'\to 3'$ WK automaton as we have shown in Theorem \ref{prop:sN1}.
\end{proof}

\begin{lemma}\label{prop:s-NS}
Let $L$ be a language strongly accepted by a sensing \textbf{NS} $5'\to 3'$ WK automaton. If it contains a nonempty word $a^n$ with some $a\in T$ and $n\in \mathbb N$, then it contains all words of $(a^n)^*$. 
\end{lemma}
\begin{proof}
Any word of the form $a^n$ can be considered as a singleton necklace. Further, as such automaton has only one state, the same computation steps as the ones lead to the acceptance of $a^n$ can be repeated if the input is longer. In this way, each word of $(a^n)^*$ is accepted, thus the language is infinite.
\end{proof}

\begin{lemma}\label{prop:sF1}
Let $L$ be a language of necklaces strongly accepted by a sensing \textbf{F1} $5'\to 3'$ WK automaton. If $L$ contains a nonempty word, then it contains one letter long word(s).
\end{lemma}
\begin{proof}
 In a sensing \textbf{F1} $5'\to 3'$ WK automaton all states are accepting, and the automaton can read exactly one letter in the first step of the computation. Thus, if it has any transition from the initial state, it will accept the one letter long word containing the letter of the transition. As every one letter long word itself is a singleton necklace, it is also strongly accepted, thus it appears in the strongly accepted necklace language. W.l.o.g., assume that there is a transition with letter $a\in T$ with the first head in $M$, i.e., $\delta(q_0,a,\lambda) \ne\emptyset$. Then $a\in L_s(M)$.   
\end{proof}

\begin{theorem}
The class $\mathcal{S}_{F1}$ properly includes the class $\mathcal{S}_{N1}$: 
$$\mathcal{S}_{N1} \subsetneq \mathcal{S}_{F1} .$$
\end{theorem}
\begin{proof} The inclusion holds by definition, as all \textbf{N1} automata are also \textbf{F1} automata. To show the properness we give an example.
Consider $M = (\{a,b\},\{q,p,r\},q,\{q,p,r\},\delta)$
with $\delta(q,a,\lambda) = \{p\} $ and $\delta(q,b,\lambda) = \{r\}$ 
  (where $\delta$ gives the empty set for any other triplets). 
It is easy to see that both the accepted and the strongly accepted language is $\{\lambda,a,b\}$ which cannot be accepted by any \textbf{N1} $5'\to 3'$ WK automaton. %
\end{proof}

\begin{theorem}
The class $\mathcal{S}_{FS}$ properly includes both of the classes $\mathcal{S}_{F1}$ and
 $\mathcal{S}_{NS}$:
$$\mathcal{S}_{F1} \subsetneq \mathcal{S}_{FS} \text{ \ \ and \ \   }\mathcal{S}_{NS} \subsetneq \mathcal{S}_{FS}.$$
\end{theorem}
\begin{proof} The inclusions hold by definition, as all \textbf{F1} automata and all \textbf{NS} automata are also \textbf{FS} automata. To show the properness we give an example.
Consider $M = (\{a,b\},\{q,p\},q,\{q,p\},\delta)$
with $\delta(q,aa,\lambda) = \{p\} $, $\delta(q,ab,\lambda) = \{p\} $  and $\delta(q,ba,\lambda) = \{p\} $ 
  (where $\delta$ gives the empty set for any other triplets).
It is easy to see that both the accepted and the strongly accepted language is $\{\lambda, aa, ab, ba\}$ includes two nonempty necklaces. This language cannot be accepted by any \textbf{F1} $5'\to 3'$ WK automaton by Lemma \ref{prop:sF1} as each of its nonempty words has length $2$. Moreover, $L_s(M)$ is a finite language containing the nonempty word $aa$, thus by Lemma  \ref{prop:s-NS} it cannot be strongly accepted by any \textbf{NS} $5'\to 3'$ WK automaton.
\end{proof}

The examples we have used so far defined regular languages. To show that the model we are considering here has a larger expressive power, we present the following example, where a non regular (and in fact, not linear context-free) language is defined by an  \textbf{F1} $5'\to 3'$ WK automaton.

\begin{example}
\begin{figure}[ht]
    \centering
        \includegraphics[scale=0.9]{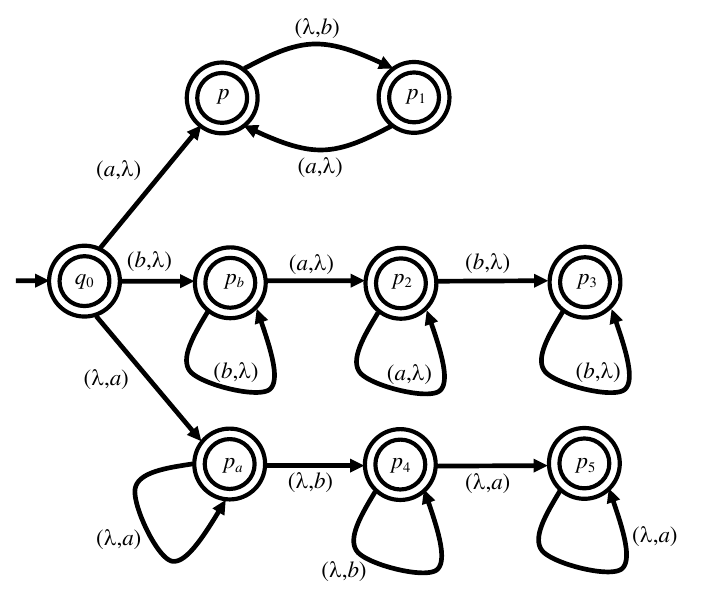}
            \caption{A  sensing \textbf{F1} $5' \rightarrow 3'$ WK automaton that is accepting a non linear context free language of necklaces in the strong mode.}
    \label{fig:F1ex}
\end{figure}
Consider the sensing \textbf{F1} $5'\to 3'$ WK automaton $M$ shown in Figure \ref{fig:F1ex}.
Depending on the first letter of the chosen conjugate, the computation follows different ways and also there is computation based on the last letter.
 If the first letter is $b$, then state $p_b$ is reached, and all continuations belong to $b^* a^* b^*$ are accepted. In this way, clearly all words of $b^*$ are also strongly accepted, as each of them is a singleton necklace.
 Whenever, the last letter of the conjugate is an $a$, there is a computation reaching $p_a$ and the computation continues accepting all words of $a^*b^*a^*$. Here all necklaces containing only $a$-s are also accepted, i.e., the elements of $a^*$ are in $L_s(M)$.  
 If the necklace contains both $a$ and $b$, then it must also be accepted when conjugate starting with $a$ and finishing with a $b$ (having the pattern $b\cdot a$ to fit to this position). However, in this case, the only computation goes from $q_0$ to $p$ and continues by using both heads and counting the number of $a$-s and $b$-s not to have a larger difference than $1$.
Thus, the strongly accepted necklace language is $\{a^*\} \cup \{b^*\} \cup \{a^nb^n\} \cup \{a^{n+1} b^n\} \cup \{ b^k a^n b^m~|~n\in\{k+m, k+m+1\} \}$. This language is not regular, moreover, it is not linear. On the other hand, it is context-free as a PDA can easily count the number of letters in each of the possible conjugates.
\end{example}

In \cite{AFL} it was proven that exactly the class $\mathcal{L}_{LIN}$ of linear context-free languages are accepted by each of the classes of (arbitrary, i.e., unrestricted) sensing $5'\to 3'$ WK automata, of sensing \textbf{S} $5'\to 3'$ WK automata and of sensing \textbf{1} $5'\to 3'$ WK automata. By considering these automata for necklaces in the strong acceptance mode, we have the following consequence on the top of the hierarchy.

\begin{proposition}
$$\mathcal{S}_* = \mathcal{S}_S = \mathcal{S}_1  \supset \mathcal{S}_{F} .$$ 
\end{proposition}

We leave open whether the hierarchy is proper or not for the pair of classes we did not show proofs. A summary of these results can also be seen in the Hasse diagram in Figure \ref{Hasse-s} in the next section.
\begin{figure}[ht]
 \centering
     \includegraphics[scale=.75]{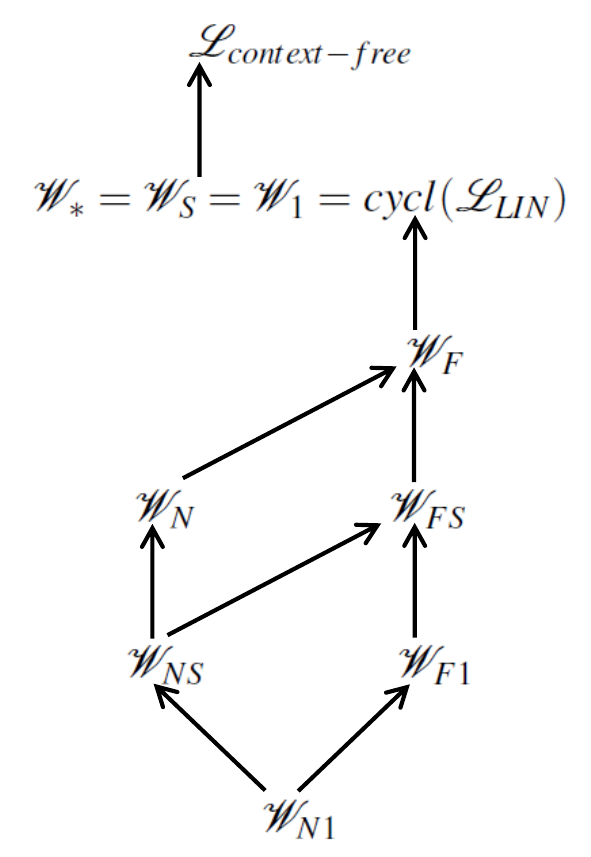}
              \caption{Hierarchy of necklace languages weakly accepted by sensing $5' \rightarrow 3'$ WK finite automata in a Hasse diagram. Each of the shown inclusions is proper.}
    \label{Hasse-w}
\end{figure}
\section{Conclusions}

Necklaces (or circular words) may represent various real word objects, e.g., DNA molecules having circular (also called) cyclic structure. In mathematics and computer science they are often modeled by the set of conjugates, i.e., linear (ordinary) words that could be the base of the cycle.
In this paper, we used WK automata to accept necklaces and necklace languages. Two acceptance modes have been investigated, if at least one of the elements of the conjugate class is accepted, then the corresponding necklace is weakly accepted, while in case all conjugates are accepted, the necklace is strongly accepted.   
Based on the various restrictions of WK automata, we established hierarchies of the accepted language classes. 
We summarize these hierarchy results obtained for necklace languages by Hasse diagrams and we also list a few open problems. 

On the first hand, a Hasse diagram shows the hierarchy of the weakly accepted classes of necklace languages in Figure \ref{Hasse-w}. 
\begin{figure}[t]
 \centering
       \includegraphics[scale=.75]{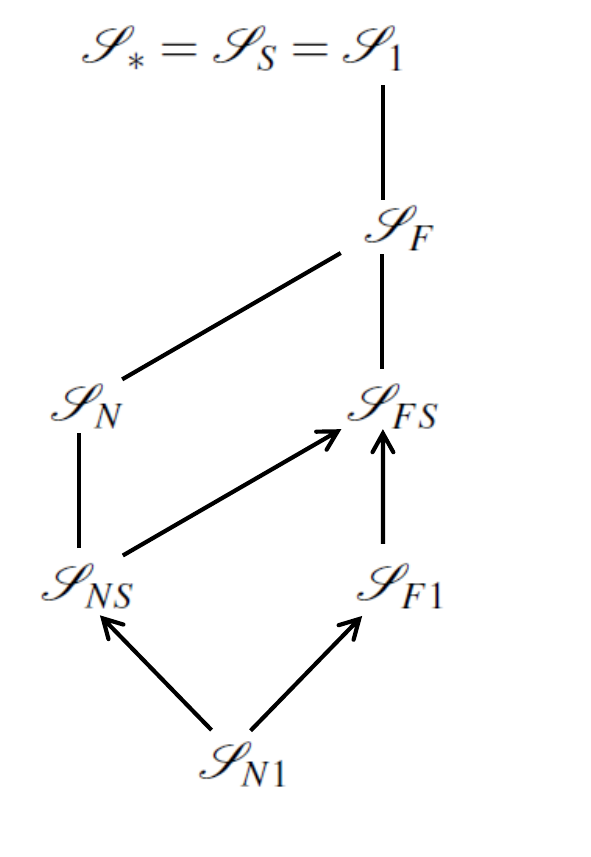}
                                \caption{Hierarchy of necklace languages strongly accepted by sensing $5' \rightarrow 3'$ WK finite automata in a Hasse diagram. Arrows represent proper inclusions, while lines represent inclusions where the properness is  left open.}
    \label{Hasse-s}
\end{figure}
On the other hand, Figure \ref{Hasse-s} shows the Hasse diagram of the language classes of the strongly accepted necklace languages. Here, some of the inclusions are trivial by definition and their properness are left open. More precisely, the relations
(equality or proper inclusion) between the following classes is open: 
\begin{itemize}
\item  $\mathcal{S}_{NS} - \mathcal{S}_N$,
\item  $\mathcal{S}_{FS} - \mathcal{S}_F$,
\item  $\mathcal{S}_N - \mathcal{S}_F$, and
\item  $\mathcal{S}_F - \mathcal{S}_*$.
\end{itemize}
 Further open problems are, e.g., the closure properties of the newly defined language classes. Relations to other families of languages, including locally testable families are also planned to be established in the near future.
%
%


\nocite{*}
\bibliographystyle{eptcs}

\begin{thebibliography}{}
%
%
\bibitem{Adleman}
Leonard M. Adleman (1994): \textit{Molecular computation of solutions to combinatorial problems, Science} 226, pp. 1021--1024, doi:\href{https://doi.org/10.1126/science.7973651}{10.1126/science.7973651}.

\bibitem{Amar} Amar, V., Putzolu, G.R. (1964): \textit{On a family of linear grammars. Inf. Control} 7(3), 283--291, doi:\href{https://doi.org/10.1016/S0019-9958(64)90294-3}{10.1016/S0019-9958(64)90294-3}.

\bibitem{genCycl}
Andreas Brandst\"adt (1981): \textit{Closure Properties of Certain Families of Formal Languages with Respect to a Generalization of Cyclic Closure. RAIRO Theor. Informatics Appl.} 15(3), pp. 233--252, doi:\href{https://doi.org/10.1051/ita/1981150302331}{10.1051/ita/1981150302331}.

\bibitem{Czeizle}
Elena Czeizler \& Eugen Czeizler (2006): \textit{A Short Survey on Watson-Crick Automata, Bulletin of the EATCS} 88, pp. 104--119.

\bibitem{Freund}
Rudolf Freund, Gheorghe  {P\u aun}, Grzegorz Rozenberg \& Arto Salomaa (1997): \textit{Watson-Crick finite automata}.
In: 
 Harvey Rubin \& David Harlan Wood, editors: \textit{DNA Based Computers, Proceedings of a DIMACS Workshop, Philadelphia, Pennsylvania, USA, June 23-25, 1997, DIMACS Series in Discrete Mathematics and
Theoretical Computer Science} 48, DIMACS/AMS, pp. 297--327, doi:\href{https://doi.org/10.1090/dimacs/048/22}{10.1090/dimacs/048/22}.

\bibitem{Hegedus}  L\'aszl\'o Heged\"us \& Benedek Nagy (2013):\textit{ Periodicity of circular words. In: WORDS 2013, Turku, Finland, TUCS Lecture Notes No.} 20 (09.2013), pp. 45--56.

\bibitem{DM} L\'aszl\'o Heged\"us \& Benedek Nagy (2016): \textit{On periodic properties of circular words. Discrete Mathematics} 339(3), pp. 1189--1197, doi:\href{https://doi.org/10.1016/j.disc.2015.10.043}{10.1016/j.disc.2015.10.043}.

\bibitem{HopUl}   John E. Hopcroft \& Jeffrey D. Ullman (1979): \textit{Introduction to Automata Theory, Languages and Computation}. Addison-Wesley, Reading, M.A.

\bibitem{NCMA-bi} 	Ondrej Kl\'{\i}ma \& Libor Pol\'ak (2011):
\textit{On Biautomata}. In Rudolf Freund, Markus Holzer, Carlo Mereghetti, Friedrich Otto, Beatrice Palano (eds.):
\textit{Third Workshop on Non-Classical Models for Automata and Applications -- NCMA 2011, Milan, Italy, July 18 - July 19, 2011.
 Proceedings. books@ocg.at} 282,
  Austrian Computer Society, pp.  153--164.

\bibitem{jumpWK} Radim Kocman, Zbynek Krivka, Alexander Meduna \& Benedek Nagy (2022): \textit{A jumping $5' \to 3'$ Watson-Crick finite automata model. Acta Informatica} 59(5), pp. 557--584, doi:\href{https://doi.org/10.1007/s00236-021-00413-x}{10.1007/s00236-021-00413-x}

\bibitem{Manfred} Manfred Kudlek (2004): \textit{On languages of cyclic words}. In: Natasa Jonoska, Gheorghe P\u aun, Grzegorz Rozenberg (eds.):
\textit{Aspects of Molecular Computing, Essays Dedicated to Tom Head on the Occasion of His 70th Birthday. Lecture Notes in Computer Science,
LNCS} 2950, pp. 278--288, doi:\href{https://doi.org/10.1007/978-3-540-24635-0_20}{10.1007/978-3-540-24635-0_20}.

\bibitem{Kuske}
Dietrich Kuske \& Peter Weigel (2004): \textit{The role of the complementarity relation in Watson-Crick automata
and sticker systems}. In:  Cristian S. Calude, Elena Calude \& Michael J. Dinneen (editors): \textit{Developments in Language Theory, DLT 2004, Lecture Notes in Computer Science,
LNCS} 3340, Springer, Berlin, Heidelberg, pp. 272--283. doi:\href{https://doi.org/10.1007/978-3-540-30550-7 23}{10.1007/978-3-540-30550-7 23}.

\bibitem{Leupold-NCMA} Peter Leupold \& Benedek Nagy (2009): $5' \rightarrow 3'$ Watson-Crick automata with several runs. In: Henning Bordihn, Rudolf Freund, Markus Holzer, Martin Kutrib, Friedrich Otto (eds.):
Workshop on Non-Classical Models for Automata and Applications - NCMA 2009, Wroclaw, Poland, August 31 - September 1, 2009. Proceedings. books@ocg.at 256, Austrian Computer Society 2009, pp.
167--180.

\bibitem{Leupold}
Peter Leupold \& Benedek Nagy (2010): $5' \rightarrow 3'$ \textit{Watson-Crick automata with several runs, Fundamenta
Informaticae} 104, pp. 71--91, doi:\href{https://doi.org/10.3233/FI-2010-336}{10.3233/FI-2010-336}.

\bibitem{R.Louk} Roussanka Loukanova (2007):
\textit{Linear context free languages}. In: Cliff B. Jones, Zhiming Liu, Jim Woodcock (eds.):
\textit{Theoretical Aspects of Computing - ICTAC 2007, 4th International Colloquium, Macau, China, September 26-28, 2007, Proceedings. Lecture Notes in Computer Science} 4711, Springer 2007, pp. 351--365, 
doi:\href{https://doi.org/10.1007/978-3-540-75292-9_24}{10.1007/978-3-540-75292-9_24}.

\bibitem{McNau-loc-test} Robert McNaughton \& Seymour Papert (1971): \textit{Counter-Free Automata}. MIT Press. 

\bibitem{DNA2008} Benedek
Nagy (2008): \textit{On $5' \rightarrow 3'$ sensing Watson-Crick finite automata}, In: Garzon M.H. \& Yan H. (eds.):
\textit{DNA Computing. DNA 2007: Selected revised papers, Lecture Notes in Computer Science, LNCS} 4848,
Springer, Berlin, Heidelberg, pp. 256--262. doi:\href{https://doi.org/10.1007/978-3-540-77962-9_27}{10.1007/978-3-540-77962-9_27}.


%
\bibitem{CiE} 
Benedek Nagy (2009): \textit{On a hierarchy of $5' \rightarrow 3'$ sensing WK finite automata languages}, In: \textit{Computaility
in Europe, CiE 2009: Mathematical Theory and Computational Practice, Abstract Booklet, Heidelberg}, pp.
266--275.

%
\bibitem{iConcept} Benedek
Nagy (2010): $5' \rightarrow 3'$ \textit{sensing Watson-Crick finite automata},  In: Gabriel Fung (ed.):
\textit{Sequence and Genome Analysis II - Methods and Applications}, pp. 39–-56, iConcept Press.

%
\bibitem{Triangle} Benedek Nagy (2012): \textit{A class of $2$-head finite automata for linear languages. 
Triangle $8$: llenguatge, literatura, computaci\'o}, 89--99.

\bibitem{Nagy2013} Benedek
Nagy (2013): \textit{On a hierarchy of $5' \rightarrow 3'$ sensing Watson-Crick finite automata languages, Journal
of Logic and Computation} 23(4), pp. 855--872, doi:\href{https://doi.org/10.1093/logcom/exr049}{10.1093/logcom/exr049}.


%
\bibitem{Eger-N} Benedek Nagy (2023): \textit{On language classes accepted by stateless $5'\to 3'$ Watson-Crick finite automata.
Annales Mathematicae et Informaticae} 58, pp. 110--120, doi:\href{https://doi.org/10.33039/ami.2023.08.004}{10.33039/ami.2023.08.004}.

%
\bibitem{RAIRO-trans} Benedek Nagy \& Zita Kov\'acs (2021): \textit{On deterministic 1-limited $5' \to 3'$ sensing Watson-Crick finite-state transducers. RAIRO Theor. Informatics Appl.} 55(5) (18 pages), doi:\href{https://doi.org/10.1051/ita/2021007}{10.1051/ita/2021007}.

\bibitem{wtl} Benedek Nagy \& Friedrich Otto (2011): \textit{Finite-State Acceptors with Translucent Letters, ICAART 2011 - 3rd International Conference on Agents and Artificial Intelligence, BILC 2011 - 1st International Workshop on AI Methods for Interdisciplinary Research in Language and Biology}, pp. 3--13, doi:\href{https://doi.org/10.5220/0003272500030013}{10.5220/0003272500030013}.

\bibitem{WKwtl} Benedek Nagy \& Friedrich Otto (2020): \textit{Linear automata with translucent letters and linear context-free trace languages. RAIRO Theor. Informatics Appl.} 54, article number 3 (23 pages), doi:\href{https://doi.org/10.1051/ita/2020002}{10.1051/ita/2020002}.

%
\bibitem{ActaInf} Benedek Nagy\& Shaghayegh Parchami (2021): \textit{On deterministic sensing $5' \to 3'$ Watson-Crick finite automata: a full hierarchy in 2detLIN, Acta Informatica} 58(3), pp. 153--175, doi:\href{https://doi.org/10.1007/s00236-019-00362-6}{10.1007/s00236-019-00362-6}.

%
\bibitem{NaCo} Benedek Nagy \& Shaghayegh Parchami (2022): $5' \to 3'$ \textit{Watson-Crick automata languages-without sensing parameter. Nat. Comput.} 21(4), pp. 679--691, doi:\href{https://doi.org/10.1007/s11047-021-09869-9}{10.1007/s11047-021-09869-9}.

\bibitem{AFL} Benedek Nagy, Shaghayegh Parchami \& Hamid-Mir-Mohammed Sadeghi (2017):
\textit{A new sensing $5'\to 3'$ Watson-Crick automata concept}. In \textit{AFL 2017:
 Proceedings 15th International Conference on Automata and Formal Languages, EPTCS} 252,
pp. 195--204, doi:\href{https://doi.org/10.4204/EPTCS.252.19}{10.4204/EPTCS.252.19}. 



%
\bibitem{UCNC} Shaghayegh Parchami, Benedek Nagy (2018): \textit{Deterministic Sensing $5' \to 3'$ Watson-Crick Automata Without Sensing Parameter}, In  Susan Stepney \& Sergey Verlan (editors): \textit{UCNC 2018: 17th International Conference on Unconventional Computation and Natural Computation, LNCS} 10867, pp. 173--187, doi:\href{https://doi.org/10.1007/978-3-319-92435-9_13}{10.1007/978-3-319-92435-9_13}.

\bibitem{Paun} Gheorghe
P\u aun, Grzegorz Rozenberg \& Arto Salomaa (2002): \textit{DNA Computing: New Computing Paradigms}.
Springer-Verlag, doi:\href{https://doi.org/10.1007/978-3-662-03563-4}{10.1007/978-3-662-03563-4}. 


\bibitem{handbook-of-NC} Grzegorz Rozenberg, Thomas B\"ack \& Joost N. Kok (2012):
\textit{Handbook of Natural Computing}. Springer, 
doi:\href{https://doi.org/10.1007/978-3-540-92910-9}{10.1007/978-3-540-92910-9}

\bibitem{Handb} Grzegorz Rozenberg \& Arto Salomaa, eds., (1997):
\textit{Handbook of Formal Languages}. Springer, doi:\href{https://doi.org/10.1007/978-3-642-59136-5}{10.1007/978-3-642-59136-5}.

\bibitem{Sempere04} Jos\'e M. Sempere (2004): \textit{A Representation Theorem for Languages Accepted by Watson-Crick
Finite Automata. Bulletin of the EATCS} 83, pp. 187--191.

\bibitem{Sempere18} Jos\'e M. Sempere (2018): \textit{On the application of Watson-Crick finite automata for the resolution
of bioinformatic problems}, In  Rudolf Freund, Michal Hospod\'ar, Galina Jir\'askov\'a \& Giovanni Pighizzini, editors: \textit{Tenth
Workshop on Non-Classical Models of Automata and Applications, NCMA} 2018, \"Osterreichische Computer Gesellschaft, 
pp. 29--30. Invited talk. 

\bibitem{Sempere}  Jos\'e M. Sempere \& P. Garc\'{\i}a (1994): \textit{A characterization of even linear languages and its application to the learning problem}. In: \textit{ICGI 1994, LNCS/LNAI} 862, pp. 38--44, doi:\href{https://doi.org/10.1007/3-540-58473-0_135}{10.1007/3-540-58473-0_135}. 


\bibitem{LocTest} Yechezkel Zalcstein (1972): \textit{Locally testable languages, Journal of Computer and System Sciences} %
6(2), pp. 151--167, doi:\href{https://doi.org/10.1016/S0022-0000(72)80020-5}{10.1016/S0022-0000(72)80020-5}.




\end{thebibliography}

\end{document}